\documentclass[conference]{IEEEtran}
\IEEEoverridecommandlockouts
\usepackage[cmex10]{amsmath}
\usepackage{amssymb,bm,amsthm,mathtools}
\usepackage{graphicx}
\usepackage[caption=false,font=footnotesize]{subfig}
\usepackage{pgfplots,pgfplotstable}
\pgfplotsset{compat=1.18}
\usepgfplotslibrary{groupplots,fillbetween}
\usepackage[dvipsnames]{xcolor}
\usepackage{microtype}
\usepackage{booktabs}
\usepackage{cite}
\usepackage{url}
\usepackage{balance}
\usepackage{algorithm,algorithmic}
\usepackage{titlesec}
\titlespacing*{\section}
  {0pt}{1.0ex plus 0.2ex minus 0.2ex}{0.5ex}
\titlespacing*{\subsection}
  {0pt}{0.8ex plus 0.2ex minus 0.2ex}{0.4ex}
\usepackage{needspace}
\DeclareMathOperator{\sgn}{sgn}
\DeclareMathOperator{\atanh}{atanh}
\newcommand{\GF}{\mathrm{GF}}
\newcommand{\Hmat}{\bm{H}}
\newcommand{\evec}{\bm{e}}
\newcommand{\svec}{\bm{s}}
\newcommand{\sres}{\tilde{\bm{s}}}
\newcommand{\Mset}{\mathcal{M}}
\newcommand{\Nset}{\mathcal{N}}
\newcommand{\Eset}{\mathcal{E}} 
\newcommand{\Sset}{\mathcal{S}} 
\newcommand{\scaleNMS}{\alpha}           
\newcommand{\alphaeff}{\alpha_{\mathrm{eff}}}  
\newcommand{\alphastar}{\alpha^{\star}}  
\newcommand{\alpharef}{\bar{\alpha}}     
\newcommand{\alphamin}{\alpha_{\min}}
\newcommand{\alphamax}{\alpha_{\max}}
\newcommand{\etaunsat}{\eta_{\mathrm{unsat}}}
\newcommand{\alphaopt}{\alpha_{\mathrm{opt}}}  
\newcommand{\tfm}{\kappa}
\newcommand{\synratio}{\gamma}
\newcommand{\Lzero}{L_0}               
\newcommand{\phifn}{\phi}
\newcommand{\Mij}{M_{i\to j}^{(\iter)}}
\newcommand{\iter}{\ell}
\newcommand{\itermax}{\ell_{\max}}
\newcommand{\eps}{\varepsilon}
\newcommand{\epszero}{\varepsilon_0}
\newcommand{\FER}{\mathrm{FER}}

\newcommand{\Prob}{\mathbb{P}}
\newtheorem{proposition}{Proposition}
\newtheoremstyle{remarkstyle}{3pt}{3pt}{\normalfont}{0pt}{\itshape}{.}{.5em}{}
\theoremstyle{remarkstyle}
\newtheorem{remark}{Remark}
\newtheoremstyle{corollarystyle}{3pt}{3pt}{\normalfont}{0pt}{\itshape}{.}{.5em}{}
\theoremstyle{corollarystyle}
\newtheorem{corollary}{Corollary}
\pgfplotsset{every axis/.style={
    label style={font=\footnotesize},
    tick label style={font=\footnotesize},
    legend style={font=\scriptsize,draw=gray!40,fill=white,
                  fill opacity=0.9,text opacity=1},
    grid=both,
    grid style={line width=0.3pt,color=gray!30},
    major grid style={line width=0.5pt,color=gray!50},
    minor tick num=1,
}}
\begin{document}
\title{Syndrome Adaptive Gain Control for Min-Sum Decoding
       of Quantum LDPC Codes}
\author{\IEEEauthorblockN{%
  Hernan~Cordova,
  Alexios~Balatsoukas-Stimming,
  Yunus~Can~G\"{u}ltekin,
  Gabriele~Liga,
  and~Alex~Alvarado}\\
\IEEEauthorblockA{Eindhoven University of Technology,
  Eindhoven, The Netherlands \\
  \{h.x.cordova, a.k.balatsoukas.stimming,
   y.c.g.gultekin, g.liga, a.alvarado\}@tue.nl}}
\maketitle
\begin{abstract}
Min-Sum~(MS) decoding is a popular low-complexity alternative to 
belief propagation (BP), retaining only the minimum incoming message magnitude during check-node (CN) processing, at the cost of systematic message magnitude overestimation. The scaled MS~(SMS) decoder compensates for this effect using a fixed scaling factor.
We propose the \emph{syndrome adaptive gain Min-Sum}~(SAGMS)
decoder for quantum low-density parity-check~(QLDPC) codes, which adapts the message gain online based on the fraction of unsatisfied stabilizers, requiring no per-code or per-noise level optimization.
We show that the scaling factor required for SMS to match belief propagation decreases with the CN degree, so any fixed scaling optimized for one degree incurs into a growing penalty as the CN degree varies. SAGMS avoids this limitation by adapting the gain during decoding.
Simulations on generalized bicycle QLDPC codes demonstrate that SAGMS matches or outperforms the frame error rate (FER) of an offline optimized SMS decoder. Moreover, SAGMS approaches BP performance and, under certain conditions outperforms it while retaining MS-level complexity.
\end{abstract}
\begin{IEEEkeywords}
Adaptive gain, Belief propagation, QLDPC, MinSum decoding, quantum error correction.
\end{IEEEkeywords}
\section{Introduction}
\label{sec:intro}
Quantum low-density parity-check (QLDPC) codes are a leading candidate for scalable quantum error correction (QEC), combining sparse stabilizer structure with minimum distance, growing polynomially in the block length for constant rate~\cite{panteleev2022asymptotically,
leverrier2022quantum,tillich2014quantum}. Decoding is commonly performed using belief propagation (BP) on a Tanner graph (TG) derived from the stabilizer parity check matrix~\cite{Mackay2004,LeiferPoulin2008}. 
The quaternary formulation of BP~(BP4), operating in a Galois field of order $4$, i.e., $\GF(4)$, allows the decoder to naturally capture degeneracy, where different Pauli error patterns produce identical syndromes~\cite{LeiferPoulin2008,Gottesman1997}. However, BP4 requires transcendental operations (e.g., $\log$, $\tanh$) at each check-node (CN) update, making hardware implementation expensive.

Min-Sum (MS) approximations replace these operations with comparisons and additions, significantly reducing complexity~\cite{chen2002near}. This reduction comes at the cost of systematic message magnitude overestimation~\cite{chen2005reduced}. Scaled MS (SMS) mitigates this effect by introducing a fixed scaling factor $\scaleNMS\in(0,1]$ to multiply CN outputs and compensate for this overestimation, but requires offline optimization of $\scaleNMS$ per code and noise level~\cite{chen2002near,chen2005reduced}.
A key property underlying the stability of BP4 on loopy graphs is its implicit gain regulation~\cite{Gallager1962}. The CN update compresses message magnitudes, preventing excessive amplification in short cycles. In contrast, SMS replaces this mechanism with $\scaleNMS$, which cannot adapt to the decoder state.
While MS and SMS are well established for classical LDPC codes and widely deployed in communication standards~\cite{RyanLin2009}, their extension to QLDPC codes is non-trivial. First, QLDPC TGs are typically denser and more irregular due to stabilizer (orthogonality) constraints and overcomplete representations~\cite{smiao_qbp_overcomplete2025}, making the best performing $\scaleNMS$ potentially highly code-dependent. 
Second, the quantum depolarizing channel couples $X$, $Y$ and $Z$ errors, requiring joint processing over $\GF(4)$ to fully exploit degeneracy. Decoders that process $X$ and $Z$ components independently lose this advantage, e.g., BP2,~\cite{LeiferPoulin2008,Gottesman1997}. Third, the noise level is often uncertain or time-varying, making offline optimization of $\scaleNMS$ unreliable in practice.

In the classical LDPC setting, adaptive scaling of Min-Sum messages has been widely studied~\cite{chen2002near,chen2005reduced}. Existing methods adjust the scaling factor using offline optimization or pretrained mappings based on quantities such as mutual information, generalized mutual information (GMI), or the signal-to-noise ratio (SNR)~\cite{xu2014variable,jung2014sanms}. Other approaches use the fraction of unsatisfied parity checks as a feedback signal during decoding~\cite{fan2013anmsa,wu2010adaptive}. However, these methods are restricted to the $\GF(2)$ setting and require per-code, per-SNR offline optimization. For the quantum setting this is an active area of research. In~\cite{chytas2025enhancedMS} authors proposed an enhanced MS using previous iteration dynamics in correction terms, while in~\cite{javed2024lowcomplexity} they proposed a linear programming-based low-complexity decoder. Both of these proposals require offline parameter calibration and treat $X$ and $Z$ errors independently, losing the GF(4)\footnote{When BP4 messages are reduced to scalar LLRs via
the scalarization of~\cite{Lai_2021,smiao_qbp_overcomplete2025},
the CN arithmetic superficially resembles binary BP but the
resulting decoder still propagates joint beliefs over all four
Paulis $\{I,X,Y,Z\}$ per qubit, exploiting degeneracy in a way
that decoders processing $X$ and $Z$ syndromes independently cannot~\cite{LeiferPoulin2008}.} degeneracy advantage. The work in~\cite{RelayBP2025} uses message memory across BP iterations to escape from trapping sets with a focus on the variable node update only. Other approaches include modified BP4 CN update rules~\cite{rigby2019modified}, TG modifications~\cite{ducrest2022stabilizer}, ordered statistics
decoding as post-processing~\cite{hillmann2025localized}, and neural BP~\cite{nachmani2018deep, smiao_qbp_overcomplete2025}.
To date, reduced-complexity decoders for QLDPC codes have been developed primarily in the binary domain or require offline optimization/training or need more complex computations.
We propose the \emph{syndrome adaptive gain Min-Sum} (SAGMS) decoder, which uses the fraction of unsatisfied stabilizers as a low-overhead feedback signal. This quantity directly reflects the distance to convergence and can be computed efficiently using simple bitwise operations\footnote{Early termination checks if any CN is unsatisfied (a single OR/XOR over the syndrome bits). SAGMS additionally computes the popcount $\synratio^{(\iter)}$, implemented as a binary adder tree, which can be pipelined with CN/VN updates at negligible throughput cost.}, with negligible impact on decoding throughput.

\textbf{Contributions:}
We propose SAGMS, a self-adaptive gain MS decoder for the $\GF(4)$
quaternary QLDPC setting, which uses the fraction of unsatisfied stabilizers as a feedback to adapt message gain at MS-level complexity. SAGMS' FER performance is comparable to a SMS decoder in matched channel conditions and exceeds it under channel mismatch and varying CN degrees. 
We prove: (i)~that the scaling factor required for SMS to match BP4 under an uniform message approximation decreases with the CN degree so any fixed scaling factor leads to a growing penalty as the CN degree varies (Proposition~\ref{prop:monotone_alpha}, ahead), and (ii)~that SAGMS steers the effective gain toward its maximum value
near convergence, independently of the CN degree and noise level.

\newlength{\mechwidth}
\newlength{\mechheight}
\setlength{\mechwidth}{0.98\linewidth}   
\setlength{\mechheight}{0.55\linewidth}  
\begin{figure*}[!htbp]
\centering
\definecolor{colBP4}  {RGB}{161, 33, 33}
\definecolor{colMS}   {RGB}{ 80, 80, 80}
\definecolor{colSMS}  {RGB}{ 69, 99,168}
\definecolor{colSAGMS}{RGB}{ 27,158,119}
\def\lw{0.80pt}
\pgfmathsetmacro{\dc}{4}
\pgfmathsetmacro{\alphaeff}{0.65}
\begin{minipage}[t]{0.495\textwidth}
\centering
\hspace{-.4cm}
\begin{tikzpicture}
\begin{axis}[
    width=\mechwidth,
    height=\mechheight,
    xlabel={\scriptsize $d_c$},
    ylabel={\scriptsize $|L_{i\to j}^{(\iter)}|$},
    ylabel near ticks,
    xmin=3, xmax=10,
    ymin=1.3, ymax=3.1,
    xtick={3,4,5,6,7,8,9,10},
    grid=major,
    grid style={dashed,gray!40},
    tick label style={font=\scriptsize},
    legend style={
        font=\scriptsize, draw=none, fill=white,
        fill opacity=0.9, text opacity=1,
        at={(0.70,0.67)}, anchor=north west,
        row sep=-2pt,
    },
    legend cell align=left,
]
\addplot[colMS, mark=square*, mark size=1.8pt, fill=white, line width=\lw,
    opacity=0.85]
    table {figures/tikz/advanced_analysis_example/minsum_output.txt};
\addlegendentry{MS}
\addplot[colSMS,mark=triangle, mark size=1.8pt, fill=white,  line width=\lw,
    opacity=0.85]
    table {figures/tikz/advanced_analysis_example/SMS_output.txt};
\addlegendentry{SMS}
\addplot[colSAGMS, mark=diamond*, mark size=2.5pt,  fill=white, line width=\lw, fill=white]
    table {figures/tikz/advanced_analysis_example/wsynd_output.txt};
\addlegendentry{SAGMS}
\addplot[colBP4, mark=o, mark size=2.0pt, line width=1.0pt]
    table {figures/tikz/advanced_analysis_example/bp4_output.txt};
\addlegendentry{BP4}
\end{axis}
\end{tikzpicture}
{\scriptsize (a)~CN output magnitude $|L_{i\to j}^{(\iter)}|$ vs.\ $d_c$.}
\end{minipage}
\hfill
\begin{minipage}[t]{0.495\textwidth}
\centering
\begin{tikzpicture}
\begin{axis}[
    width=\mechwidth,
    height=\mechheight,
    xlabel={\scriptsize $\tfm$},
    ylabel={\scriptsize $T(\tfm)$},
    ylabel near ticks,
    xmin=0, xmax=3,
    ymin=0, ymax=3.2,
    domain=0.05:3,
    samples=20,
    grid=major,
    grid style={dashed,gray!40},
    tick label style={font=\scriptsize},
    legend style={
        font=\scriptsize, draw=none, fill=white,
        fill opacity=0.9, text opacity=1,
        at={(0.05,0.99)}, anchor=north west,
        row sep=-2pt,
    },
    legend cell align=left,
]
\addplot[colMS, mark=square*, mark size=1.8pt,  fill=white, line width=\lw,
    opacity=0.85]
    {x};
\addlegendentry{MS}
\addplot[colSMS, mark=triangle, mark size=1.8pt,  fill=white, line width=\lw,
    opacity=0.85]
    {0.80*x};
\addlegendentry{SMS}
\addplot[colSAGMS, mark=diamond*, mark size=2.5pt,  fill=white, line width=\lw, fill=white]
    {\alphaeff * x};
\addlegendentry{SAGMS}
\addplot[colBP4, mark=o, mark size=2.0pt,  fill=white, line width=1.0pt]
 {ln((1 + (tanh(x/2))^(\dc-1)) /
        max(1 - (tanh(x/2))^(\dc-1), 1e-10))};
\addlegendentry{BP4}

\end{axis}
\end{tikzpicture}
{\scriptsize
(b)~Transfer function $T(\tfm)$ (see~\eqref{eq:transfer_all})}
\end{minipage}
\caption{\scriptsize
(a)~Magnitude bias of the MS rule relative to BP4,
measured through the CN output magnitude $|L_{i\to j}^{(\iter)}|$
vs.\ CN degree~$d_c$.
(b)~Transfer function (TF) $T(\tfm)$ showing SAGMS provides a
$1^{\text{st}}$-order approximation to the BP4 TF. Illustration (using $d_c{=}4$, $\alpha{=}0.85$, $\alpha_{\mathrm{eff}}{=}0.65$).
}
\label{fig:decoder_mechanism}
\end{figure*}

\section{System Model}
\label{sec:model}
An $[\![n,k]\!]$ QLDPC code has a PCM
$\Hmat\in\GF(4)^{m\times n}$ satisfying
$\Hmat\bm{\Lambda}\Hmat^{\mathsf{T}}=\bm{0}$ over
$\GF(2)$ where $\bm{\Lambda}$ denotes the symplectic form used to define the trace inner product over $\GF(4)$~\cite{Gottesman1997,calderbank1996good}\footnote{\textbf{Notation:}
Scalars, vectors, and matrices are denoted by italic
lowercase ($a$), bold lowercase ($\bm{a}$), and bold
uppercase ($\bm{A}$) letters, respectively.
Sets are written in calligraphic uppercase ($\mathcal{A}$). 
Set of edges $(j,i)$ of the TG are denoted via $\Eset$. $\Mset(i)$/$\Nset(j)$: check-node/variable-node~(VN)
neighbors of CN~$i$ and VN~$j$;
$d_c=|\Mset(i)|$: CN degree; $\Sset$: set of stabilizers checks;
$\iter$: iteration index; $\itermax$: maximum iterations;
$\eps$: true depolarization probability;
$\epszero$: assumed depolarization probability;
$\Lzero$: channel prior log-likelihood ratio~(LLR) defined as $\Lzero \triangleq \ln\!\left(\frac{1-\epszero}{\epszero/3}\right)$;
$\synratio^{(\iter)}$: syndrome ratio;
$\alphaeff^{(\iter)}$: adaptive effective gain;
$\alphastar$: BP4 matching ratio (defined
in~\eqref{eq:alphastar});
$\alpharef$: fixed scaling optimized for a reference CN
degree (used in Proposition~\ref{prop:monotone_alpha});
$\etaunsat>1$: multiplicative boost applied to $\alphaeff$
at unsatisfied CNs (see~\eqref{eq:alphaeff}).}.
Here, $\GF(4)=\{0,1,\bar{\omega},\omega\}$ with elements
identifying Pauli errors $\{I,X,Y,Z\}$~\cite{Gottesman1997,calderbank1996good}.
Each $\GF(4)$ element is mapped via the symplectic isomorphism $\phi_s:\GF(4)\to\GF(2)^2$ after which $\Hmat$ is represented in $\GF(2)^{m\times 2n}$ and $\Hmat\bm{\Lambda}\Hmat^{\mathsf{T}}$ is computed using the binary symplectic inner products between rows. Equivalently this corresponds to evaluating the trace inner product in $\GF(4)$~\cite{Gottesman1997}.
The codes considered in this work are from the generalized bicycle~(GB)
family~\cite{panteleev2022asymptotically,bravyi2024high}.
Under the depolarizing channel, each qubit
$j\in\{1,\ldots,n\}$ suffers an independent Pauli error
$E_j\in\{I,X,Y,Z\}$ with $\Prob\{E_j{=}I\}{=}1\!-\!\eps$ and
$\Prob\{E_j{=}e\}{=}\eps/3$ for $e\in\{X,Y,Z\}$.
The channel is memoryless, so the $n$-qubit error
$\boldsymbol{E}=(E_1,\dots,E_n)$ has joint distribution
$\Prob\{\boldsymbol{E}{=}\evec\}{=}\prod_{j=1}^n\Prob\{E_j{=}e_j\}$.
The syndrome bit at CN $i$ is defined as
$s_i{=}\langle H_i,\evec\rangle_{\mathrm{tr}}\in\{0,1\}$,
where $H_i$ denotes the $i$-th row of $\Hmat$ and
$\langle\cdot,\cdot\rangle_{\mathrm{tr}}$ is the trace inner product
over $\GF(4)$~\cite{Gottesman1997}.
The residual syndrome at iteration $\iter$ is
\begin{equation}\vspace{-2pt}
\sres^{(\iter)} \triangleq
\svec \oplus \Hmat\hat{\evec}^{(\iter)} \;\text{ over }\GF(2),
\label{eq:resid}
\end{equation}\vspace{-2pt}
where $\hat{\evec}^{(\iter)}$ is the tentative error
estimate with each component
$\hat{e}_j^{(\iter)}\in\{0,1,\bar{\omega},\omega\}
\leftrightarrow\{I,X,Z,Y\}$, identical in alphabet to
the entries of $\Hmat$.
The matrix-vector product
$\Hmat\hat{\evec}^{(\iter)}$ uses $\GF(4)$ field
arithmetic (multiplication $\leftrightarrow$ Pauli
composition modulo global phase, and addition
$\leftrightarrow$ bitwise XOR of the symplectic pairs)
and is then projected per-coordinate contributions in inner product to $\GF(2)$ via $\phi_s$~\cite{Gottesman1997,calderbank1998quantum},
yielding a nonzero syndrome bit precisely when a
stabilizer element anticommutes with the corresponding
estimated error. Decoding succeeds when $\sres^{(\iter)}{=}\bm{0}$.
We denote by $\tilde{s}_i^{(\iter)}$ the $i$-th component
of $\sres^{(\iter)}$, with $\tilde{s}_i^{(\iter)}{=}1$ if CN~$i$ is unsatisfied
at iteration~$\iter$ and $\tilde{s}_i^{(\iter)}{=}0$ otherwise.
An overcomplete stabilizer representation~($m{>}n\!-\!k$) is
used, as in~\cite{smiao_qbp_overcomplete2025,
cumitini2024optimalsingleshotdecodingquantum}.
The redundant rows improve FER performance
at the cost of denser short cycles in the TG.
Throughout, we write $[\![n,k]\!]$ to denote a code and state $m$ separately (when required)
to avoid confusion with the minimum distance notation $[\![n,k,d]\!]$.

\textit{Why $\GF(4)$ rather than $\GF(2)$?}
$\GF(2)$-based decoders process $X$ and $Z$ syndromes independently,
losing two key properties: correction decisions are uncoordinated
(unmatched $Y$-syndrome assignments cause the combined correction to
fail~\cite{LeiferPoulin2008}), and degeneracy through $Y$-type
stabilizer products cannot be exploited~\cite{Gottesman1997}.
The $\GF(4)$ formulation handles both natively by propagating joint
beliefs over $\{I,X,Y,Z\}$ per qubit~\cite{LeiferPoulin2008,Mackay2004}.

\textit{BP4 decoder.}
Messages are initialized with $L_{j\to i}^{(1)}{=}\Lzero$
for all edges $(j,i)\in\Eset$, and CN to VN messages at
$\iter{=}1$ are set to zero (no prior CN output available).
The VN update rule at iteration $\iter\geq1$ is
\begin{equation}\vspace{-2pt}
L_{j\to i}^{(\iter)}=\Lzero+
\sum_{k\in\Nset(j)\setminus i}L_{k\to j}^{(\iter-1)}.
\label{eq:bp4_vn}
\end{equation}\vspace{-2pt}
The CN update in the $\phifn$-domain is
\begin{equation}\vspace{-2pt}
L_{i\to j}^{(\iter)} {=}
(-1)^{s_i}
\prod_{k\in\Mset(i)\setminus j}\sgn\!\bigl(L_{k\to i}^{(\iter)}\bigr)\,
\phifn^{-1}\!\Bigl(\sum_{k\in\Mset(i)\setminus j}
\phifn\!\bigl(|L_{k\to i}^{(\iter)}|\bigr)\Bigr)
\label{eq:bp4_cn}
\end{equation}\vspace{-2pt}
where $\phifn(x){=}-\log\tanh(x/2)$ is the Gallager
$\phifn$-function~\cite{Gallager1962}.
The nonlinear $\phifn^{-1}$ provides implicit gain
regulation: since $\phifn$ maps large inputs toward zero
and $\phifn^{-1}{=}\phifn$, the composition
$\phifn^{-1}(\sum_k\phifn(\cdot))$ is always
upper bounded by the minimum incoming magnitude,
preventing runaway amplification along short cycles in the TG.
This is the key property any good MS approximation must preserve.

\textit{MS and SMS decoders.}
The MS CN update replaces $\phifn^{-1}(\cdot)$ with the
minimum incoming magnitude~$\Mij$, i.e.,
\begin{equation}
L_{i\to j}^{(\iter)} =
(-1)^{s_i}
\prod_{k\in\Mset(i)\setminus j}
\sgn\!\bigl(L_{k\to i}^{(\iter)}\bigr)
\Mij.
\label{eq:ms_cn}
\end{equation}
SMS has the same update multiplied by $\scaleNMS\in(0,1]$.
The CN transfer functions (i.e., output magnitude vs.\
minimum input $\tfm=\Mij$) are shown in Fig.\ref{fig:decoder_mechanism}(b) and are given by:
\begin{align}
T_{\mathrm{BP4}}(\tfm) &{=} \phifn^{-1}\!\bigl((d_c-1)\phifn(\tfm)\bigr),\text{   } T_{\mathrm{MS}}(\tfm) {=} \tfm, \notag\\
T_{\mathrm{SMS}}(\tfm) &{=}\alpha\,\tfm, \text{   }
T_{\mathrm{SAGMS}}(\tfm) {=} \alpha_{\mathrm{eff}}^{(\ell)}\,\tfm.
\label{eq:transfer_all}
\end{align}
where all $d_c\!-\!1$ incoming messages are set equal to $\tfm$ (uniform message approximation (UMA), exact at $\iter{=}1$) so each BP4 curve corresponds to a specific $d_c$ value and the single input $\tfm$. At $\iter{>}1$, message distributions widen and the curves shift but the qualitative overestimation of MS relative to BP4 persists.
As $\phifn^{-1}$ is strictly decreasing,
$T_{\mathrm{BP4}}(\tfm)\leq \tfm,\forall \tfm>0$,
confirming that MS systematically overestimates the BP4 output.
\section{SAGMS Decoder}
\label{sec:sagms}
\subsection{SAGMS Parameters Definitions}
\textit{Syndrome ratio.} Let $\Sset$ denote the set of stabilizer checks (CNs)
in the TG with cardinality $|\Sset|$.
At iteration $\iter$, the syndrome ratio is
\begin{equation}\vspace{-2pt}
\synratio^{(\iter)} =
\|\sres^{(\iter)}\|_0 / |\Sset|,
\label{eq:synratio}
\end{equation}\vspace{-2pt}
with $\synratio^{(\iter)}{=}0$ when all stabilizers are
satisfied and $\synratio^{(\iter)}{=}1$ when none are. 
During successful decoding, $\gamma^{(\iter)}$
decreases in general toward zero as the tentative
error estimate $\hat{\evec}^{(\iter)}$ converges toward
a valid stabilizer coset representative.\\
\noindent
\textit{Adaptive effective gain.}
The adaptive effective gain\footnote{Note that $\alphaeff^{(\iter)}$ in~\eqref{eq:alphaeff} depends on $\tilde{s}_i^{(\iter)}$ and is therefore CN-specific. With some abuse of notation, we suppress the index $i$ for notational brevity.} at CN~$i$ and iteration~$\iter$ is
\begin{equation}\vspace{-2pt}
\alphaeff^{(\iter)} =
\bigl[\alphamax - (\alphamax-\alphamin)\synratio^{(\iter)}\bigr]
\cdot
\begin{cases}
\etaunsat, & \tilde{s}_i^{(\iter)}=1, \\
1,         & \tilde{s}_i^{(\iter)}=0,
\end{cases}
\label{eq:alphaeff}
\end{equation}\vspace{-2pt}
where $\alphamin,\alphamax\in(0,1]$ and $\etaunsat>1$.
The stability constraint $\alphamax\etaunsat\leq1$
ensures $\alphaeff^{(\iter)}\leq1$ at all iterations,
so SAGMS never outputs a larger message magnitude than
unscaled MS ($T_{\mathrm{MS}}(\tfm)=\tfm$). Since MS already
overestimates the BP4 output (Section~\ref{sec:model}),
allowing $\alphaeff^{(\iter)}>1$ would cause runaway
amplification on short cycles, which is precisely
the instability that gain scaling is designed to prevent.
As $\synratio^{(\iter)}\to0$ the gain rises toward
$\alphamax$, delivering maximum corrective push near
convergence.
At high $\synratio^{(\iter)}$, it falls to
$\alphamin$, suppressing short cycle oscillation.
The parameter $\etaunsat>1$ is a fixed multiplicative (moderate) boost applied to unsatisfied CNs. It is selected such that $\alphamax\etaunsat\leq1$ ensuring SAGMS never exceeds MS scaling and we interpret it as a local correction beyond the global gain for CNs still far from their individual local optimal points. In practice, $\etaunsat$ is set to the largest value satisfying
$\alphamax\etaunsat\leq1$, directing maximum additional
corrective push to unsatisfied CNs without exceeding pure MS.\\
\noindent
\textit{SAGMS update rule.}
The VN update
\begin{equation}
L_{i\to j}^{(\iter)} =
(-1)^{s_i}
\prod_{k\in\Mset(i)\setminus j}
\sgn\!\bigl(L_{k\to i}^{(\iter)}\bigr)
\cdot\alphaeff^{(\iter)}\,\Mij
\label{eq:sagms_cn}
\end{equation}
is the same as BP4.
Algorithm~\ref{alg:sagms} summarizes the decoder.
\subsection{Benefits of the Syndrome Ratio as Feedback Signal}
\label{sec:gamma_motivation}

The syndrome ratio $\synratio^{(\iter)}$ provides a natural low-complexity signal for online gain adaptation in QLDPC decoding due to the following properties.

\textit{(i) It is the lowest cost convergence-state observable.}
The residual syndrome $\sres^{(\iter)}$ is already computed in each iteration for early termination. While termination only requires detecting whether $\|\sres^{(\iter)}\|_0{=}0$, SAGMS additionally computes the normalized count $\synratio^{(\iter)}=\|\sres^{(\iter)}\|_0/|\Sset|$. This can be implemented via a popcount operation using a binary adder tree of $O(m)$ single bit additions, which can be pipelined and executed in parallel with CN and VN updates, incurring negligible overhead.
No other decoder state signal (e.g., message variance,
LLR magnitude, or extrinsic information
transfer) is available at comparable cost.

\textit{(ii) It tracks the true convergence distance.}
For a converging decoder, $\synratio^{(\iter)}$ decreases in general toward zero as the tentative estimate approaches a valid stabilizer coset representative. It therefore provides a direct and noise-averaged measure of the distance to convergence. In contrast, message magnitudes evolve non-monotonically on loopy graphs due to short cycle correlations and are unreliable as convergence indicators.

\textit{(iii) It provides per-CN local information via $\tilde{s}_i^{(\iter)}$.}
In addition to the global ratio, the individual syndrome bits $\tilde{s}_i^{(\iter)}\!\in\!\{0,1\}$ identify which CNs remain unsatisfied. This enables the per-CN gain differentiation in~\eqref{eq:alphaeff}, where unsatisfied CNs receive additional corrective boost. No extra computation is required since $\tilde{s}_i^{(\iter)}$ is directly available from the syndrome evaluation.
Unlike prior adaptive schemes in the $\GF(2)$ setting~\cite{fan2013anmsa,wu2010adaptive,xu2014variable}, SAGMS requires neither offline calibration nor explicit noise estimation.
LUT-based adaptive methods~\cite{fan2013anmsa,wu2010adaptive,zhao2005adaptive}
use only the global syndrome fraction and do not
differentiate per-CN gain.

\begin{algorithm}[t]
\caption{\small SAGMS Decoder}
\label{alg:sagms}
\small
\begin{algorithmic}[1]
\REQUIRE $\Hmat$, $\svec$, $\Lzero$, $\itermax$
\hfill\COMMENT{$\alphamin,\alphamax,\etaunsat$ from~\eqref{eq:alphaeff}}
\STATE Initialize $L_{j\to i}^{(0)}\!\leftarrow\!\Lzero$
  for all $(j,i)\in\Eset$
\FOR{$\iter=1$ \TO $\itermax$}
\STATE  Compute hard decision $\hat{e}_j \leftarrow
    \arg\min_{e\in\{I,X,Y,Z\}}\bigl[\Lzero +
    \sum_{k\in\Nset(j)}L_{k\to j}^{(\iter-1)}\bigr]$
    for all VN~$j$; then $\sres^{(\iter)}\leftarrow
    \svec\oplus\Hmat\hat{\evec}^{(\iter)}$
  \IF{$\sres^{(\iter)}=\bm{0}$} \RETURN $\hat{\evec}^{(\iter)}$ \hfill(success)
  \ENDIF
  \STATE Compute $\synratio^{(\iter)}\leftarrow
    \|\sres^{(\iter)}\|_0/|\Sset|$
  \STATE For each CN $i$: compute $\alphaeff^{(\iter)}$
    via~\eqref{eq:alphaeff}; update $L_{i\to j}^{(\iter)}$
    via~\eqref{eq:sagms_cn}
  \STATE For each VN $j$: update $L_{j\to i}^{(\iter)}$
    via the BP4 VN rule
\ENDFOR
\RETURN $\hat{\evec}^{(\itermax)}$ \hfill(failure)
\end{algorithmic}
\end{algorithm}
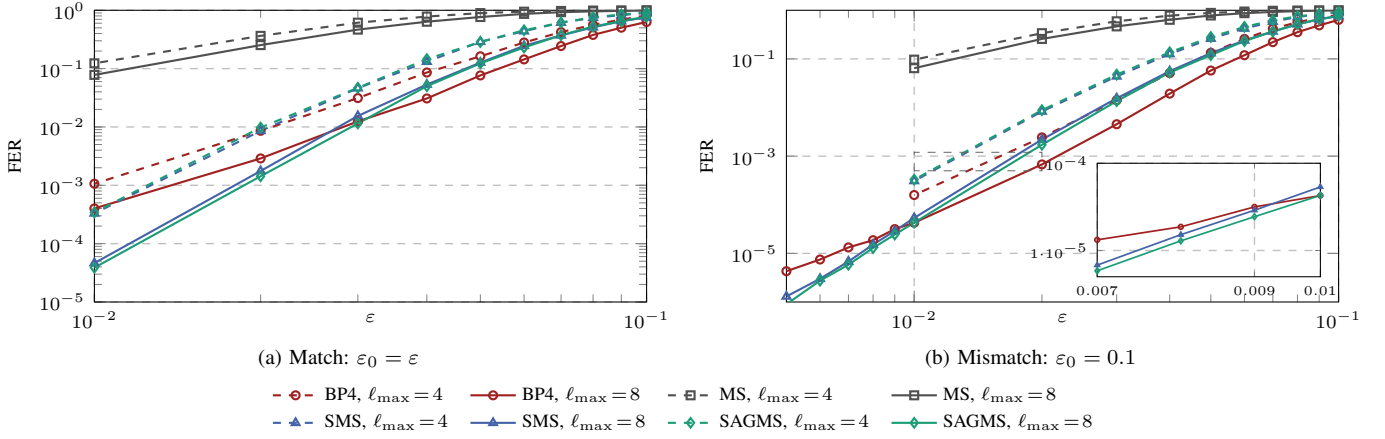
\begin{figure*}[!t]
\centering
\definecolor{colBP4}  {RGB}{161, 33, 33}
\definecolor{colMS}   {RGB}{ 80, 80, 80}
\definecolor{colSMS}  {RGB}{ 69, 99,168}
\definecolor{colSAGMS}{RGB}{ 27,158,119}
\def\lw{0.80pt}
\subfloat[Match: $\eps_0=\eps$\label{fig:FER_match}]{%
\begin{tikzpicture}
\begin{axis}[
    width  = 0.49\textwidth,
    height = 0.30\textwidth,
    xmode=log, ymode=log,
    xmin=1e-2, xmax=1e-1,
    ymin=1e-5,  ymax=1,
    xlabel={\scriptsize $\eps$},
    xlabel style={
        at={(axis description cs:0.5,1.01)},
        anchor=south,
        yshift=-122pt
    },
    ylabel={\scriptsize FER},
    ylabel near ticks,
    grid=major,
    grid style={dashed,gray!40},
    tick label style={font=\scriptsize},
    yticklabel style={/pgf/number format/sci,
                      /pgf/number format/precision=0},
    legend to name=FERlegend,
    legend columns=4,
    legend style={
        font=\scriptsize,
        draw=none,
        fill=none,
        /tikz/every even column/.append style={column sep=0.6em},
        row sep=0pt,
    },
    legend cell align=left,
]
\addplot[color=colBP4, dashed, line width=\lw,
    mark=o, mark options={solid,scale=0.75,fill=white,line width=\lw}]
    file[]{figures/tikz/FER_decoders/GB_n126_k28_m126_BP4_iter4_match.txt};
\addlegendentry{BP4,~$\ell_{\max}\!=\!4$}
\addplot[color=colBP4, solid, line width=\lw,
    mark=o, mark options={solid,scale=0.75,fill=white,line width=\lw}]
    file[]{figures/tikz/FER_decoders/GB_n126_k28_m126_BP4_iter8_match.txt};
\addlegendentry{BP4,~$\ell_{\max}\!=\!8$}
\addplot[color=colMS, dashed, line width=\lw,
    mark=square, mark options={solid,scale=0.75,fill=white,line width=\lw}]
    file[]{figures/tikz/FER_decoders/GB_n126_k28_m126_MS_iter4_match.txt};
\addlegendentry{MS,~$\ell_{\max}\!=\!4$}
\addplot[color=colMS, solid, line width=\lw,
    mark=square, mark options={solid,scale=0.75,fill=white,line width=\lw}]
    file[]{figures/tikz/FER_decoders/GB_n126_k28_m126_MS_iter8_match.txt};
\addlegendentry{MS,~$\ell_{\max}\!=\!8$}
\addplot[color=colSMS, dashed, line width=\lw,
    mark=triangle, mark options={solid,scale=0.85,fill=white,line width=\lw}]
    file[]{figures/tikz/FER_decoders/GB_n126_k28_m126_MS4FixedPv2_iter4_match.txt};
\addlegendentry{SMS,~$\ell_{\max}\!=\!4$}
\addplot[color=colSMS, solid, line width=\lw,
    mark=triangle, mark options={solid,scale=0.85,fill=white,line width=0.8pt}]
    file[]{figures/tikz/FER_decoders/GB_n126_k28_m126_MS4FixedPv2_iter8_match.txt};
\addlegendentry{SMS,~$\ell_{\max}\!=\!8$}
\addplot[color=colSAGMS, dashed, line width=\lw,
    mark=diamond, mark options={solid,scale=0.90,fill=white,line width=0.8pt}]
    file[]{figures/tikz/FER_decoders/GB_n126_k28_m126_SAGMS_iter4_match.txt};
\addlegendentry{SAGMS,~$\ell_{\max}\!=\!4$}
\addplot[color=colSAGMS, solid, line width=\lw,
    mark=diamond, mark options={solid,scale=0.90,fill=white,line width=0.8pt}]
    file[]{figures/tikz/FER_decoders/GB_n126_k28_m126_SAGMS_iter8_match.txt};
\addlegendentry{SAGMS,~$\ell_{\max}\!=\!8$}
\end{axis}
\end{tikzpicture}%
}\hfill
\subfloat[Mismatch: $\eps_0=0.1$\label{fig:FER_mismatch}]{%
\begin{tikzpicture}
\begin{axis}[
    name=mainmismatch,
    width  = 0.49\textwidth,
    height = 0.30\textwidth,
    xmode=log, ymode=log,
    xmin=5e-3, xmax=1e-1,
    ymin=1e-6,  ymax=1,
    xlabel={\scriptsize $\eps$},
    xlabel style={
        at={(axis description cs:0.5,1.01)},
        anchor=south,
        yshift=-122pt,
    },
    ylabel={\scriptsize FER},
    ylabel near ticks,
    grid=major,
    grid style={dashed,gray!40},
    tick label style={font=\scriptsize},
    yticklabel style={/pgf/number format/sci,
                      /pgf/number format/precision=0},
]
\addplot[color=colBP4, dashed, line width=\lw,
    mark=o, mark options={solid,scale=0.75,fill=white,line width=0.8pt}]
    file[]{figures/tikz/FER_decoders/GB_n126_k28_m126_BP4_iter4_mismatch_ep0_0p10.txt};
\addplot[color=colBP4, solid, line width=\lw,
    mark=o, mark options={solid,scale=0.75,fill=white,line width=0.8pt}]
    file[]{figures/tikz/FER_decoders/GB_n126_k28_m126_BP4_iter8_mismatch_ep0_0p10.txt};
\addplot[color=colMS, dashed, line width=\lw,
    mark=square, mark options={solid,scale=0.75,fill=white,line width=0.8pt}]
    file[]{figures/tikz/FER_decoders/GB_n126_k28_m126_MS_iter4_mismatch_ep0_0p10.txt};
\addplot[color=colMS, solid, line width=\lw,
    mark=square, mark options={solid,scale=0.75,fill=white,line width=0.8pt}]
    file[]{figures/tikz/FER_decoders/GB_n126_k28_m126_MS_iter8_mismatch_ep0_0p10.txt};
\addplot[color=colSMS, dashed, line width=\lw,
    mark=triangle, mark options={solid,scale=0.85,fill=white,line width=0.8pt}]
    file[]{figures/tikz/FER_decoders/GB_n126_k28_m126_MS4FixedPv2_iter4_mismatch_ep0_0p10.txt};
\addplot[color=colSMS, solid, line width=\lw,
    mark=triangle, mark options={solid,scale=0.85,fill=white,line width=0.8pt}]
    file[]{figures/tikz/FER_decoders/GB_n126_k28_m126_MS4FixedPv2_iter8_mismatch_ep0_0p10.txt};
\addplot[color=colSAGMS, dashed, line width=\lw,
    mark=diamond, mark options={solid,scale=0.90,fill=white,line width=0.8pt}]
    file[]{figures/tikz/FER_decoders/GB_n126_k28_m126_SAGMS_iter4_mismatch_ep0_0p10.txt};
\addplot[color=colSAGMS, solid, line width=\lw,
    mark=diamond, mark options={solid,scale=0.90,fill=white,line width=0.8pt}]
    file[]{figures/tikz/FER_decoders/GB_n126_k28_m126_SAGMS_iter8_mismatch_ep0_0p10.txt};
\draw[black!60!white, very thin, dashed]
    (axis cs:1e-2, 5e-4) rectangle (axis cs:2e-2, 1.2e-3);
\end{axis}
\begin{axis}[
    name=insetmismatch,
    at={(mainmismatch.south west)},
    anchor=south west,
    xshift=117pt,
    yshift=9.5pt,
    width  = 0.25\textwidth,
    height = 0.17\textwidth,
    xmode=log, ymode=log,
    xmin=7e-3,   xmax=1e-2,
    ymin=5e-6,   ymax=1e-4,
    xtick={0.009, 0.007, 0.01},
    xticklabels={$0.009$, $0.007$, $0.01$},
    ytick={1e-4, 1e-5},
    yticklabels={$10^{-4}$, $1{\cdot}10^{-5}$},
    tick label style={font=\tiny},
    scaled ticks=false,
    tick align=inside,
    grid=major,
    grid style={dashed, gray!60},
    axis background/.style={fill=white},
    axis line style={thin, black},
    major tick length=2pt,
]
\addplot[color=colBP4, solid, line width=0.65pt,
    mark=o, mark options={solid,scale=0.45,fill=colBP4}]
    file[]{figures/tikz/FER_decoders/GB_n126_k28_m126_BP4_iter8_mismatch_ep0_0p10.txt};
\addplot[color=colSMS, solid, line width=0.65pt,
    mark=triangle, mark options={solid,scale=0.50,fill=colSMS}]
    file[]{figures/tikz/FER_decoders/GB_n126_k28_m126_MS4FixedPv2_iter8_mismatch_ep0_0p10.txt};
\addplot[color=colSAGMS, solid, line width=0.65pt,
    mark=diamond, mark options={solid,scale=0.55,fill=colSAGMS}]
    file[]{figures/tikz/FER_decoders/GB_n126_k28_m126_SAGMS_iter8_mismatch_ep0_0p10.txt};
\end{axis}
\end{tikzpicture}%
}%

\vspace{2pt}
\centerline{\pgfplotslegendfromname{FERlegend}}

\caption{\scriptsize
FER vs.\ $\eps$ for BP4, MS, SMS ($\scaleNMS\!=\!0.50$),
and SAGMS ($\alphamin\!=\!0.30$, $\alphamax\!=\!0.50$,
$\etaunsat\!=\!1.10$) on the GB $[\![126,28]\!]$ QLDPC code ($m{=}126$).
Dashed: $\itermax\!=\!4$; solid: $\itermax\!=\!8$.
Inset in~(b) zooming in $\eps\!\in\![7{\cdot}10^{-3},\,1{\times}10^{-2}]$
for $\itermax\!=\!8$, highlighting the SAGMS crossover below BP4.
}
\label{fig:FER_results}
\end{figure*}
\section{Analytical Results}
\label{sec:analysis}
\subsection{BP4 Matching Ratio (under UMA) and Scalability Penalty}
\label{sec:alphastar}
\begin{remark}[\textit{SMS scaling factor from BP4/MS output}]
\label{rem:alphastar}
Under UMA (all $d_c\!-\!1$ magnitudes equal $\Lzero$, exact at
$\iter{=}1$), the MS~CN output equals $\Lzero$ and the BP4~CN
output equals $\phifn^{-1}\!\bigl((d_c\!-\!1)\,\phifn(\Lzero)\bigr)$.
The scaling factor that makes the SMS~CN output equal the BP4~CN
output at $\iter\!=\!1$ is:
\begin{equation}\vspace{-2pt}
\alphastar(\Lzero) {=}
\frac{1}{\Lzero}2\,\atanh\!\bigl(\tanh(\Lzero/2)^{d_c-1}\bigr)
\;\!\approx\!\;
1 - \frac{\ln(d_c-1)}{\Lzero},
\label{eq:alphastar}
\end{equation}\vspace{-2pt}
where the approximation uses $\phifn(x)\approx2e^{-x}$ for $x\gg1$.
Since $\phifn^{-1}$ compresses, $\alphastar<1$: unscaled MS
overestimates the BP4 output by a factor $1/\alphastar$.
\end{remark}
\begin{proposition}[Monotone decrease and scalability penalty]
\label{prop:monotone_alpha}
The BP4-matching ratio $\alphastar(\Lzero,d_c)$
in~\eqref{eq:alphastar} is strictly decreasing in $d_c$:
$\alphastar(\Lzero,d_c\!+\!1)\!<\!\alphastar(\Lzero,d_c)$
for all $d_c\!\geq\!2$.
A fixed $\alpharef$ optimized for a reference CN degree
$d_c^{\mathrm{ref}}$ becomes mismatched when applied to a
different degree $d_c^{\mathrm{new}}$, incurring the
scalability penalty:
\begin{equation}\vspace{-2pt}
\Delta\alpha =
\frac{\ln\bigl[(d_c^{\mathrm{new}}-1)/(d_c^{\mathrm{ref}}-1)\bigr]}
      {\Lzero},
\label{eq:mismatch}
\end{equation}\vspace{-2pt}
which grows unbounded as $d_c^{\mathrm{new}}\to\infty$.
\end{proposition}
\begin{proof}
The monotone decrease follows from the analytical extension
$\tilde{\alpha}(x){=}2\,\atanh(\tanh(\Lzero/2)^{x-1})/\Lzero$,
whose derivative $\mathrm{d}\tilde{\alpha}/\mathrm{d}x{=}
{-}1/[(x-1)\Lzero]<0$ for all $x>1$. Hence
$\alphastar(\Lzero,d_c\!+\!1)<\alphastar(\Lzero,d_c)$
for all integers $d_c\!\geq\!2$.
The penalty $\Delta\alpha$ follows by subtracting~\eqref{eq:alphastar}
evaluated at $d_c^{\mathrm{ref}}$ and $d_c^{\mathrm{new}}$.
\end{proof}
\begin{remark}[\textit{Consequence for SAGMS}]
SAGMS avoids the growing mismatch of
Proposition~\ref{prop:monotone_alpha} structurally, as its gain
is driven by $\synratio^{(\iter)}$ and requires neither $d_c$
nor $\Lzero$ explicitly.
\end{remark}
\subsection{Gain Ceiling and Self-Calibration}
\label{sec:ceiling}
\begin{corollary}[Gain ceiling and near-convergence self-calibration]
\label{cor:alphamax}
At $\iter{=}1$ all messages equal $\Lzero$, so the UMA is exact
and $\alphamax{=}\alphastar(\Lzero)$ from~\eqref{eq:alphastar}.
At $\iter{>}1$ message magnitudes grow and their distribution
widens, so the true per-message $\alphastar$ is no smaller than
at $\iter{=}1$, thus $\alphamax{=}\alphastar(\Lzero)$ remains a
conservative ceiling throughout decoding.
As $\synratio^{(\iter)}\!\to\!0$,~\eqref{eq:alphaeff} gives
$\alphaeff^{(\iter)}\!\to\!\alphamax\etaunsat$ for unsatisfied
CNs and $\alphaeff^{(\iter)}\!\to\!\alphamax$ for satisfied CNs,
independently of $d_c$ and $\Lzero$. SAGMS self-calibrates
toward the analytically optimal $\alphastar(\Lzero)$.
\end{corollary}
\subsection{Gain Law Design: Linear Approximation}
\label{sec:gainlaw}
The specific linear dependence of $\alphaeff^{(\iter)}$ on
$\synratio^{(\iter)}$ in~\eqref{eq:alphaeff} is not arbitrary.
We show next that it arises as the first-order Taylor approximation
of the optimal per-iteration gain around the convergence point
$\mu^\star$, with $\synratio^{(\iter)}$ used as a proxy for
the normalized convergence deficit
$(\mu^\star - \mu^{(\iter)})/\mu^\star$.

\textit{Fixed-point condition.}
Let $\mu^{(\iter)}$ denote the mean extrinsic message magnitude
at iteration $\iter$. Under the density evolution~(DE) framework
for regular LDPC codes~\cite{RichardsonUrbanke2008}, the
update $\mu^{(\iter+1)}{=}f(\mu^{(\iter)})$ converges to a
fixed point $\mu^\star$ satisfying $f(\mu^\star){=}\mu^\star$.
For the SMS decoder with gain $\scaleNMS$, the CN update
contributes a magnitude bias of
\begin{equation}\vspace{-2pt}
\mu_{\mathrm{MS}}^{(\iter+1)}
= \scaleNMS \cdot g\!\bigl(\mu^{(\iter)}, d_c\bigr),
\label{eq:de_ms}
\end{equation}\vspace{-2pt}
where $g(\mu, d_c) = \mathbb{E}[\min_{k=1}^{d_c-1}|L_k|]$
is the expected minimum of $d_c{-}1$ i.i.d.\ messages with
mean $\mu$.
We define the \emph{optimal gain at mean magnitude $\mu$} as
\begin{equation}\vspace{-2pt}
\alphaopt(\mu) :=
\frac{T_{\mathrm{BP4}}(\mu)}{g(\mu,d_c)},
\label{eq:alphaopt}
\end{equation}\vspace{-2pt}
i.e., the value of $\scaleNMS$ making the SMS~CN output equal
the BP4~CN output when all messages have mean~$\mu$.
Note that $\alphastar(\Lzero)$ in~\eqref{eq:alphastar} is the
special case $\alphaopt(\Lzero)$ at $\iter{=}1$ under UMA,

\textit{Adaptive gain as trajectory correction.}
Away from $\mu^\star$, $\synratio^{(\iter)}$ measures the
fractional distance to convergence, where
$\synratio^{(\iter)}\!\approx\!0$ means near the fixed point
and $\synratio^{(\iter)}\!\approx\!1$ means far from it.
Expanding $\alphaopt(\mu^{(\iter)})$ to first order around
$\mu^\star$:
\begin{equation}\vspace{-2pt}
\alphaopt\!\bigl(\mu^{(\iter)}\bigr) \approx
\alphaopt(\mu^\star)
+ \frac{\partial\alphaopt}{\partial\mu}\biggr|_{\mu^\star}
\!\!\!\cdot\bigl(\mu^{(\iter)} - \mu^\star\bigr).
\label{eq:taylor}
\end{equation}\vspace{-2pt}
Since $\mu^{(\iter)}<\mu^\star$ when decoding has not converged
and $\synratio^{(\iter)}$ increases as the difference
$(\mu^\star-\mu^{(\iter)})$ grows, substituting
$\synratio^{(\iter)}$ as a proxy for
$(\mu^\star-\mu^{(\iter)})/\mu^\star$ and setting
$\alphamax{=}\alphaopt(\mu^\star){=}\alphastar(\Lzero)$
(Corollary~\ref{cor:alphamax}) and $\alphamin$ as the boundary
value at $\synratio{=}1$ yields
\begin{equation}\vspace{-2pt}
\alphaopt\!\bigl(\mu^{(\iter)}\bigr)
\approx \alphamax {-} (\alphamax {-} \alphamin)\,\synratio^{(\iter)},
\label{eq:linear_approx}
\end{equation}\vspace{-2pt}
which is exactly the global ramp in~\eqref{eq:alphaeff}.
The parameters $\alphamin$ and $\alphamax$ are the boundary
values at maximum violation ($\synratio\!\to\!1$) and at
convergence ($\synratio\!\to\!0$).

\textit{Local CN differentiation.}
The per-CN boost $\etaunsat$ in~\eqref{eq:alphaeff} extends
this first-order approximation by directing additional corrective
push to unsatisfied CNs ($\tilde{s}_i^{(\iter)}{=}1$), which
have not reached their local fixed point, at no additional
computational cost. The stability constraint
$\alphamax\etaunsat\!\leq\!1$ ensures $\alphaeff^{(\iter)}\!\leq\!1$
at all iterations. Together,~\eqref{eq:alphaeff} implements a
two-level linear approximation to the ideal per-CN, per-iteration
gain $\alphaopt(\mu_i^{(\iter)})$ using only
$\tilde{s}_i^{(\iter)}$ and $\synratio^{(\iter)}$, both
available at negligible cost.
\section{Complexity and Simulations}
\label{sec:sims}
\subsection{Complexity}
Using the weighted operation count of~\cite{chen2002near,RyanLin2009}
with $(a,b,c,d)=(1,1,1,10)$ to denote the costs of additions, multiplications,
comparisons, and transcendental evaluations, respectively, 
the \emph{per CN update complexity} (producing one output message per edge) is given by:
\begin{equation}
C_{\mathrm{BP4}} = 22d_c-13,\quad C_{\mathrm{SMS}} = 2d_c-2,\quad C_{\mathrm{SAGMS}} = 2d_c+3.
\label{eq:complexity}
\end{equation}
These values match the entries in Table~\ref{tab:decoder_models}. SAGMS adds three additions and one precomputed
multiplication over SMS, as can be verified in Table~\ref{tab:decoder_models}.
At $d_c=10$, $C_{\mathrm{BP4}}=207$,
$C_{\mathrm{SMS}}=18$, and $C_{\mathrm{SAGMS}}=23$.
As observed in~\eqref{eq:complexity}, the gap between BP4 and SAGMS grows linearly with $d_c$, while the overhead of SAGMS over SMS remains constant. 
\begin{table}[t]
\centering
\caption{Decoder Models and Operation Counts per CN Update}
\label{tab:decoder_models}
\renewcommand{\arraystretch}{0.90}
\setlength{\tabcolsep}{4pt}
\resizebox{\columnwidth}{!}{
\begin{tabular}{@{}lllcccc@{}}
\toprule
\textbf{Decoder}
  & \textbf{$T(u)$}
  & \textbf{Add} & \textbf{Mul}
  & \textbf{Cmp} & \textbf{Trans.} & \textbf{Adapt?} \\
\midrule
BP4   
  & $\phi^{-1}\!\bigl((d_c{-}1)\phi(u)\bigr)$
  & $d_c\!-\!2$ & $d_c\!-\!1$& $0$ & 2$d_c\!-\!1$ & no \\
MS
  & $u$
  & $0$ & $d_c\!-\!1$ & $d_c\!-\!2$ & $0$ & no \\
\textbf{SAGMS}
  & $\alphaeff^{(\ell)}\,u$
  & $3$ & $d_c\!+\!1$& $d_c\!-\!1$ & $0$ & yes \\
\bottomrule
\end{tabular}
}
\end{table}
\subsection{Setup}
Simulations use the $[\![126,28]\!]$ GB code\footnote{To validate Proposition~\ref{prop:monotone_alpha}, we simulate another GB code $[\![126,20]\!]$ with OS ($m{=}126$, $d_c{=}16$) only under the mismatch scenario ($\epszero{=}0.10$).} with $m{=}126$,
$d_c{=}d_v{=}10$~\cite{panteleev2021asymptotically,
smiao_qbp_overcomplete2025}.
Monte Carlo trials run until $e_t{=}500$ frame failures are
recorded per noise level, over at most $2{\times}10^7$ frames,
giving a $95\%$ Wilson CI~\cite{wilson1927probable} of
${\leq}{\pm}9\%$ relative half-width per point,
independently of the FER level.
SMS uses $\scaleNMS{=}0.50$, optimized offline for this code and noise range.
SAGMS uses $\alphamin{=}0.30$, $\alphamax{=}0.50$, and
$\etaunsat{=}1.10$, satisfying $\alphamax\etaunsat{=}0.55\leq1$,
The online gain control via $\synratio^{(\iter)}$ reduces
$\alphaeff^{(\iter)}$ toward $\alphamin$ when the decoder
is far from convergence and raises it toward $\alphamax$
as $\synratio^{(\iter)}\!\to\!0$.
The ratio $\alphamin/\alphamax{=}0.60$ sets the adaptation
range, and $\etaunsat{=}1.10$ provides a moderate per-CN
corrective boost at unsatisfied CNs.

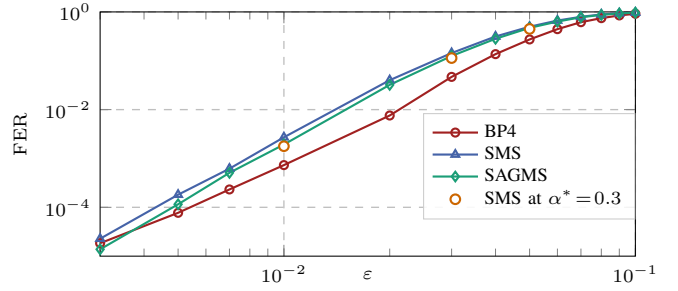
\begin{figure}[t]
\centering
\definecolor{colBP4}  {RGB}{161, 33, 33}
\definecolor{colSMS}  {RGB}{ 69, 99,168}
\definecolor{colSAGMS}{RGB}{ 27,158,119}
\definecolor{colOpt}  {RGB}{204,102,  0}   
\def\lw{0.80pt}
\begin{tikzpicture}
\begin{axis}[
    width  = 0.99\columnwidth,
    height = 0.55\columnwidth,
    xmode  = log, ymode = log,
    xmin   = 3e-3, xmax = 1e-1,
    ymin   = 1e-5, ymax = 1,
    xlabel = {\scriptsize $\eps$},
    ylabel = {\scriptsize $\FER$},
    ylabel near ticks,
    xlabel style={
    at={(axis description cs:0.5,1.01)},
    anchor=south,
    yshift=-105pt
    },
    grid   = major,
    grid style={dashed,gray!40},
    tick label style={font=\scriptsize},
    yticklabel style={/pgf/number format/sci,
                      /pgf/number format/precision=0},
    legend style={
        font=\scriptsize, draw, fill=white,
        fill opacity=0.9, text opacity=1,
        at={(1,0.15)}, anchor=south east,
        row sep=-2pt,
    },
    legend cell align=left,
]
\addplot[color=colBP4, solid, line width=\lw,
    mark=o,
    mark options={solid,scale=0.75,fill=white,line width=0.8pt}]
    file[]{figures/tikz/FER_newcode/GB_n126_k20_m126_BP4_iter8_mismatch_ep0_0p10.txt};
\addlegendentry{BP4}
\addplot[color=colSMS, solid, line width=\lw,
    mark=triangle,
    mark options={solid,scale=0.85,fill=white,line width=0.8pt}]
    file[]{figures/tikz/FER_newcode/GB_n126_k20_m126_SMS_a0p50_iter8_mismatch_ep0_0p10.txt};
\addlegendentry{SMS}
\addplot[color=colSAGMS, solid, line width=\lw,
    mark=diamond,
    mark options={solid,scale=0.90,fill=colSAGMS,line width=0.8pt}]
    file[]{figures/tikz/FER_newcode/GB_n126_k20_m126_SAGMS_iter8_mismatch_ep0_0p10.txt};
\addlegendentry{SAGMS}
\addplot[color=colOpt, only marks, line width=\lw,
    mark=*,
    mark options={solid,scale=0.90,fill=white}]
    file[]{figures/tikz/FER_newcode/GB_n126_k20_m126_SMS_iter8_mismatch_alpha0.3_ep0_0p10.txt};
\addlegendentry{SMS at $\alpha^{*}\!=\!0.3$}

\end{axis}
\end{tikzpicture}
\caption{\scriptsize
$\FER$ vs.\ $\eps$ for BP4, SMS ($\alpha\!=\!0.50$), and SAGMS
on the GB $[\![126,20]\!]$ code ($m\!=\!126$, $d_c\!=\!16$),
$\itermax\!=\!8$, mismatch ($\epszero\!=\!0.1$).
Match case omitted: both SMS and SAGMS outperform BP4, see Fig.~\ref{fig:FER_results}.
Orange points: SMS with offline-optimised $\alpha^\star\!=\!0.30$
(3 points, not connected).
SAGMS uses same parameters as for $d_c\!=\!10$ and achieves similar performance than SMS (offline optimized for this code, in orange).
}
\label{fig:FER_dc16}
\end{figure}

\subsection{Results}
Fig.~\ref{fig:FER_results} shows $\FER$ versus $\eps$ for
BP4, MS, SMS, and SAGMS.
MS performs poorly across the entire range due to systematic
magnitude overestimation, with $\FER{>}12\%$ at $\eps{=}0.01$,
consistent with the $1/\alphastar{\approx}2.2{\times}$ CN
output overestimation at $d_c{=}10$ predicted
by~\eqref{eq:alphastar}.

For $\epszero=\eps$ (Match case), at $\itermax=8$, BP4 achieves the
lowest $\FER$ for $\eps\geq0.04$.
Below $\eps{\approx}0.03$, SAGMS produces a steeper waterfall
slope\footnote{The regime where FER drops steeply with
decreasing noise level.} and crosses below both BP4 and SMS;
SMS crosses below BP4 between $\eps{=}0.03$ and $\eps{=}0.02$.
At $\eps{=}0.01$, SAGMS reaches $\FER{=}3.9{\times}10^{-5}$
and SMS reaches $4.6{\times}10^{-5}$, both approximately
$10{\times}$ lower than BP4 ($4.0{\times}10^{-4}$).
The two-sample log-ratio CI on the SAGMS/SMS ratio is
$[0.74,\,0.95]$~\cite{wilson1927probable}, confirming a
statistically significant advantage for SAGMS at this point.
At $\eps{=}0.01$ and $\itermax{=}4$, SAGMS achieves
$\FER{=}3.4{\times}10^{-4}$, comparable to BP4 at
$\itermax{=}8$ ($4.0{\times}10^{-4}$), reaching equivalent
performance at half the iteration budget.
The mechanism underlying the BP4 crossing is as follows. At low noise, $\Lzero$ is large and
$\alphastar(\Lzero)\to1$ from~\eqref{eq:alphastar}, so the SMS output already closely approximates BP4 and needs only minimal scaling.
The $\phifn^{-1}$ compression of BP4, calibrated to attenuate
messages by $\alphastar{<}1$, becomes increasingly excessive
relative to messages already well-calibrated by the large
channel prior. SMS and SAGMS, whose gain approaches the BP4-matching value in
this regime, pass larger magnitudes per CN update and converge
faster.
As $\eps$ increases, $\Lzero$ decreases, $\alphastar(\Lzero)$
decreases, and BP4 compression becomes necessary to counteract
the intrinsic overestimation of MS messages, restoring BP4's
advantage for $\eps\!\gtrsim\!0.04$.
Under fixed $\epszero{=}0.10$ (mismatch), all decoders benefit from the
regularization effect identified in prior work~\cite{smiao_qbp_overcomplete2025},
with the matched LLR now above the optimal initialization for low noise.
BP4 retains implicit gain regulation via $\phifn^{-1}$
and outperforms both MS variants for $\eps\!\gtrsim\!0.03$.
With $\itermax=8$, SAGMS crosses below BP4 at
$\eps\approx0.009$--$0.010$ and achieves lower $\FER$ below this
point; at $\eps{=}0.01$, both reach $\FER\approx4.3{\times}10^{-5}$.
This crossover does not occur for SMS, which remains above BP4
throughout the mismatch regime, confirming that the gain is specific
to the online adaptation of SAGMS.
SAGMS also outperforms SMS for $\eps\lesssim0.03$.

SMS and SAGMS achieve nearly identical FER in the matched case and both outperform BP4 in this scenario. However, under channel mismatch, SAGMS outperforms BP4 in low noise regime
($\eps\approx0.009$--$0.010$) with $\itermax=8$. 
For the $[\![126,20]\!]$ GB
code ($m{=}126$, $d_c{=}16$), under the mismatch scenario, $\alphastar\approx1-\ln(15)/3.30\approx0.18$,
substantially below $0.33$ (the value for $d_c{=}10$) at UMA.
A fixed $\scaleNMS{=}0.50$ therefore exceeds $\alphastar(\Lzero,16)$
by $0.32$, and the degree mismatch quantified
by~\eqref{eq:mismatch} gives $\Delta\alpha\approx0.15$, both
predicting SMS degradation.
The match case is omitted: both SMS and SAGMS outperform BP4 under matched conditions on this code.
SAGMS uses the same values as for $d_c{=}10$.
Fig.~\ref{fig:FER_dc16} shows that SAGMS
outperforms SMS at all simulated $\eps$.
The orange markers show the FER achieved by SMS
with the offline-optimised scaling $\alpha^\star{=}0.30$ for this code under mismatch.
The FER curves in Fig.~\ref{fig:FER_dc16} show that the gap
between SAGMS and BP4 narrows consistently as $\eps$ decreases.
Extending this trend to a few points in the lower noise regime, we observe a crossover $\eps\approx0.003$--$0.004$, analogous to the~$\eps\!\approx\!0.009$--$0.010$ crossing for $d_c{=}10$.
\section{Conclusions}
\label{sec:conclusion}
We proposed SAGMS, a low-complexity adaptive MS decoder for
QLDPC codes with online gain driven by the syndrome ratio, which requires no offline parameter optimization while being independent of $d_c$ and the noise level.
We showed that the
BP4-matching ratio $\alphastar$ decreases
monotonically with $d_c$ (for $d_c\!\ge\!2)$, so any static~$\scaleNMS$ incurs a growing mismatch as codes
scale. SAGMS avoids this penalty by design.
Simulations on the $[\![126,28]\!]$ GB code with $\itermax=8$ 
confirm that SAGMS outperforms both offline optimized SMS and BP4 under channel match conditions from $\eps\le0.03$ onward.
Under channel mismatch, SAGMS outperforms SMS for $\eps\lesssim0.03$ and outperforms BP4 for $\eps\lesssim0.010$, with the crossover observed between $\eps=0.010$ and $\eps=0.009$.
Validation on a second GB code $[\![126,20]\!]$,
under channel mismatch, confirms that SAGMS outperforms SMS
($\scaleNMS{=}0.50$) across the full noise range and achieves similar
performace than the offline-optimised SMS ($\alpha^\star{=}0.30$) without
any per-code per noise level optimization, confirming the advantage of adaptive scaling as $d_c$ varies.
\balance
\bibliographystyle{IEEEtran}
\bibliography{references_paper}
\end{document}